\documentclass{article}
\usepackage[utf8]{inputenc}
\usepackage{tikz}
\usepackage{xcolor}
\usepackage{weishi-macro}
\usepackage[T1]{fontenc}
\usepackage{authblk}

\title{GAN-MC: a Variance Reduction Tool for Derivatives Pricing}
\author[1]{Weishi Wang\thanks{weishi@uchicago.edu}}
\affil[1]{\textit{Department of Statistics, The University of Chicago}}

\begin{document}

\maketitle
\section{Abstract}
We propose a parameter-free model for estimating the price or valuation of financial derivatives like options, forwards and futures using non-supervised learning networks and Monte Carlo. Although some arbitrage-based pricing formula performs greatly on derivatives pricing like Black-Scholes on option pricing, generative model-based Monte Carlo estimation(GAN-MC) will be more accurate and holds more generalizability when lack of training samples on derivatives, underlying asset's price dynamics are unknown or the no-arbitrage conditions can not be solved analytically. We analyze the variance reduction feature of our model and to validate the potential value of the pricing model, we collect real world market derivatives data and show that our model outperforms other arbitrage-based pricing models and non-parametric machine learning models. For comparison, we estimate the price of derivatives using Black-Scholes model, ordinary least squares, radial basis function networks, multilayer perception regression, projection pursuit regression and Monte Carlo only models.

\section{Introduction}
Financial derivatives are used for risk management, hedging, speculation, and arbitrage. Better understanding of pricing of derivatives could help traders better hedge against risk, and the price of derivatives could reflect the fluctuations on the underlying assets. Much of the success and growth of the market for options and other derivatives securities should be traced to the seminal work by Black and Scholes~\cite{black1976pricing} and Merton~\cite{merton1973theory}. They introduced the closed-form option pricing formulas through no-arbitrage conditions and dynamic hedging arguments. Such celebrated Black-Scholes and Merton formulas have been well generalized, extended and applied to various securities. Nicole, Monique and Steven~\cite{karoui1998robustness} provide conditions under which the Black–Scholes formula is robust with respect to a misspecification of volatility. Wu~\cite{wu2004pricing} introduces the fuzzy set theory to the Black–Scholes formula, which attaches belief degree on the European option. Marcin~\cite{magdziarz2009black} introduces a subdiffusive geometric Brownian motion to underlying asset prices' dynamics and tests the pricing model on prices of European option. Carmona and Valdo~\cite{carmona2005generalizing} generalize the Black-Scholes formula in all dimensions by approximate formulas and provide lower and upper bounds for hedging of multivariate contingent claims. Moreover, while closed-form expressions are not available in some generalizations and extensions, pricing formulas may still take effects numerically.
\par However, the derivation of the pricing formula via the hedging or no-arbitrage approach, either analytically or numerically, highly depends on the particular parametric form of the underlying asset's price dynamics. Thus the misspecification of the stochastic process will lead to system pricing and hedging errors for derivatives related to this price. Therefore, previous parametric pricing methods are closely tied to the ability of capturing the dynamics of underlying asset prices' process.
\par In this paper, we creatively introduce the generative model-based Monte Carlo estimation(GAN-MC) for derivatives pricing and hedging. We will not assume any specific dynamics on the underlying asset prices. We only treat the asset prices as simple multivariate random variables and try to approximate its distribution by a neural network. Then we get the pricing formula from derivatives' definition and Monte Carlo estimation for the statistical stability. Compared to the previous non-parametric pricing approach like Hutchinson~\cite{hutchinson1994nonparametric}, our model relies less on derivatives' regime and more stable with the advantage of Monte Carlo.
\par In order to better capture the dynamics of underlying asset prices through non-parametric approach, we introduce generative adversarial nets(GAN)~\cite{goodfellow2014generative} to approximate underlying asset prices' distribution. GAN is the framework for estimating generative models via an adversarial process. The celebrated neural network model has been widely generalized and extended. Zhang, Goodfellow, etc~\cite{zhang2019self} propose the Self-Attention Generative Adversarial Network which allows attention-driven and long-range dependency modeling for generation tasks. Mehdi and Simon~\cite{mirza2014conditional} introduce Conditional GAN, which uses additional information to direct the data generation process. Chen, Lin, etc~\cite{chen2020dggan} propose the Depth-image Guided GAN which adds some architectural constraints to network and generates realistic depth maps conditioned on input image. Martin and Soumith~\cite{arjovsky2017wasserstein} introduce the Wasserstein GAN which stabilize the training process by replacing the original metric by Wasserstein-1 distance. Chen, Duan, etc~\cite{chen2016infogan} propose InfoGAN, an information-theoretic extension to the generative adversarial net which is able to learn disentangled representations in a completely unsupervised manner by attempting to make conditional learned automatically.

\par
Monte Carlo could be used for option pricing under different underlying asset prices' dynamics assumptions. The original approach is raised by Boyle~\cite{boyle1977options}, he uses risk neutrality to obtain equilibrium rate of return on underlying assets and uses Monte Carlo to improve efficiency of estimation. Fu and Hu~\cite{fu1995sensitivity} introduce techniques for the sensitivity analysis of Monte Carlo option pricing and they propose an approach for the pricing of options with early exercise features. Birge~\cite{birge1995quasi} introduces quasi Monte Carlo sequences which have order of magnitude better asymptotic error rate and such sequences could be used in option pricing. Mark~\cite{broadie1997enhanced} presents several enhancements to reduce the bias as well as variance of Monte Carlo estimators and improve the efficiency of the branching based estimators. Poirot and Tankov~\cite{poirot2006monte} relate the underlying asset prices to the tempered stable (also known as CGMY) processes and under an appropriate equivalent probability measure a tempered stable process becomes a stable process, thus provide a fast Monte Carlo algorithm for European option pricing.

\par The attention on training on biased datasets is increasing in recent days. The work from Yo-whan Kim, Samarth Mishra, etc~\cite{kim2022transferable} raise the idea of pre-training on synthetic video data. Compared with directly training on real video clips data, the model will perform better on downstream tasks when pre-training on the synthetic or biased datasets. Our GAN-MC model's success on derivatives pricing could be analogous to their success. Estimation based on synthetic or fake underlying asset prices outperforms non-parametric models directly trained on real derivatives prices.

\subsection{Our Contributions}
We summarize the major contributions of our paper as follows:
\begin{itemize}
    \item We first introduce the generative model-based Monte Carlo estimation for derivatives pricing. We assume that the underlying asset prices follow multivariate random variable distribution and use GAN to approximate the distribution. Then we use Monte Carlo estimation to get pricing formula for each derivative: option, forward and futures.
    \item We get the consistent estimators for prices of different derivatives theoretically and validate the accuracy of our pricing algorithms on real market data. Compared with arbitrage-based pricing formula like Black-Scholes formula, non-parametric pricing models like radial basis function networks, multilayer perception regression, and projection pursuit regression, and some simple models like linear regression and Monte Carlo only models, our GAN-MC pricing model always reaches state-of-the-art on real market data.
\end{itemize}
\par We organize this paper as follows: In Section 3, we define the problem setup and some assumptions for data representation and training. In Section 4, we introduce our GAN-MC model for derivatives pricing, including option, forward and futures. We cover European call option, European put option, American call option and American put option in option pricing. And we cover commodity and equity for forward and futures pricing. We still prove the variance of estimator will decrease with the increase of generated sample size in this section. In Section 5, we conduct experiments to test the generated sample and test the accuracy of our algorithms. Compared with other models, our GAN-MC always reaches state-of-the-art on real market option, futures and forward data.

\section{Problem setup}
\begin{figure}[htbp]
\centering
\begin{tikzpicture}
\tikzset{
box/.style ={
rectangle, 
rounded corners =5pt, 
minimum width =30pt, 
minimum height =30pt, 
inner sep=5pt, 
draw=blue }};

\tikzset{
box2/.style ={
rectangle, 
rounded corners =5pt, 
minimum width =50pt, 
minimum height =20pt, 
inner sep=5pt, 
draw=blue }
}

\node[box] (1) at(0,0) {$\Omega_{1}$};
\node[box] (2) at(0,3) {$\mathbb{R}^{n}$};
\node  at(0,4.5) {$\times$};
\node[box] (4) at(0,6) {$\mathbb{R}^{m_{1}}$};
\node (5) at(3,0) {$(\Omega_{1},F_{1},\mathbb{P}_{Y})$};
\node[box] (6) at(3,3) {$\mathbb{R}^{m\times m}$};
\node[box] (7) at(6,3) {$\mathbb{R}^{m_{2}}$};
\node[box2] (8) at(6,6) {\{0,1\} or $\mathbb{R}$};
\node[box] (9) at(3,6) {$\Omega_{2}$};

\node (11) at(1.5,6) {$(\Omega_{2},F_{2},\mathbb{P}_{\gamma})$};
\node  at(4.5,3) {$\times$};

\node  at(0.2,1.5) {$Z$};
\node  at(2.8,1.5) {$Y=G_{\theta g}\circ Z$};
\node  at(1.5,3.2) {$G$};
\node  at(4,4.5) {$D$};

\node  at(2.8,4.5) {$\gamma$};
\draw[->] (1)--(2);
\draw[->] (2)--(6);
\draw[->] (1)--(6);
\draw[->] (6)--(8);
\draw[->] (9)--(6);

\end{tikzpicture}
\caption{GAN structure} \label{fig:GAN}
\end{figure}
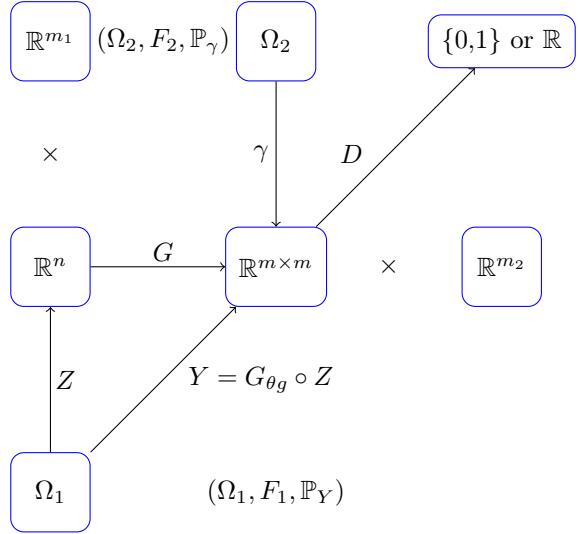
The pricing methods for financial derivatives are highly based on the prediction of underlying assets, like stock prices. The basic idea of Monte Carlo for derivatives pricing is about generating fake stock price samples. Similarly here, if we denote the stock price vector $\mathbf{S}_{t,T}=\left( S_{t},S_{t+1},\cdots,S_{t+T-1} \right)^{\top}\in\mathbb{R}_+^{T}$ as a multivariate random variable from time $t$ to time $t+T-1$, where $\mathbf{S}_{t,T} : \Omega \to \mathbb{R}_+^{T}$ is a measurable function mapping from sample space to $T$-dimensional positive real space. Then one stock price vector on real stock market $\mathbf{s}_{t,T}=\left(s_{t},s_{t+1},\cdots,s_{t+T-1} \right)$ would be a realization of $\mathbf{S}_{t,T}$.

\par For the generative adversarial nets, as shown in figure~\ref{fig:GAN}, we denote $G:\mathbb{R}^n \times \mathbb{R}^{m_1} \to \mathbb{R}^{m\times m}$ as the generator mapping, where $m_1$ is the number of parameters for generator network, $n$ is the dimension of random noise $\mathbf{Z}$, we denote $D:\mathbb{R}^T\times \mathbb{R}^{m_2} \to \{0,1\}$ as the discriminator mapping, where $m_2$ is the number of parameters for discriminator network, and we denote $\gamma:\Omega_{2} \rightarrow \mathbb{R}^{m \times m}$ is the real distribution mapping to image space. Then the training loss of GAN could be described as $\min_{G}\max_{D}M(\mathbb{P}_{Y},\mathbb{P}_{\gamma})$, where $M(\cdot,\cdot)$ is the metric between two measurable functions, $\mathbb{P}_{\gamma}$ is the probability measure of random variable $\gamma$, $\mathbb{P}_Y$ is the probability measure of random noise after mapping of generator $G$.

We first consider the dynamic structure of $\mathbf{S}_{t,T}$.

\begin{assump}[Date Independence]
\label{assump:date-indep}
If $T$ is a relative large number, 
the distribution of $\mathbf{S}_{t,T}$ does not depend on initial time point $t$, or equivalently the distribution of $\mathbf{S}_{t,T}$ could be written as the distribution of $\mathbf{S}_{T}$.
\end{assump}
\par Assumption~\ref{assump:date-indep} states that if $T$ is large, the high dimension distribution would be complex enough to include all the volatility on stock market within a period of time, like three years. GAN is a powerful tool to learn a specific high dimensional distribution from training set. Thus we can treat the stock price data as 1-D real distribution($\gamma$ in figure~\ref{fig:GAN}) under such assumption.  

\begin{assump}[Covariance Rank]
\label{assump:cov-rank}
The matrix rank of $\textrm{Cov}\left(\mathbf{S}_T\right)$ should not be too small.

\end{assump}
\par The successful training of GAN models requires that we should not train the generator on highly-correlated samples to avoid loss collapse. For a given $T$, The rank of $ \textrm{Cov}\left(\mathbf{S}_T\right)$ ranges from $1$ to $T$. Assumption~\ref{assump:cov-rank} states that the rank should not be close to 1, which guarantees the covariance structure of stock market data is complex enough for GAN to learn and the training process would not collapse at most time.

\section{Methodology}
In this section, we introduce the main algorithm, Generative Adversarial Nets-Monte Carlo(GAN-MC) model, for derivatives pricing like option, futures and forward pricing.

\subsection{GAN-MC for Option Pricing}\label{sec:option}
The theory of option pricing estimates the value of an option contract by assigning a price, known as premium, based on the calculated probability that the contract will finish in the money(ITM) at expiration. The option pricing theory provides an evaluation of an option's fair value, and the accurate pricing model could help traders better incorporate the option value into their strategies.
\par 
If we denote a specific underlying stock price data as $\mathbf{s}_{1,n}=(s_1,s_2,\cdots,s_{n})$ which starts at day 1 with length $n-1$, the continuously compounded annual risk-free interest rate as $r$, the value of a call or put option on a stock that pays no dividend as $C$ or $P$, the exercise price for the given option as $X$, the proportion of a year before the option expires as $T_0$ and $\Delta t$ is the time unit. And $T$ is a fixed parameter which satisfies Assumption~\ref{assump:date-indep} and Assumption~\ref{assump:cov-rank}. We set $N_1$ as the sample size threshold for training, and $N_2$ as the size of fake data generation. While $\alpha$ is a proportion parameter which controls the weights that different $\mathbf{s}_{t,T}$ contributes for the generation. Then algorithm~\ref{alg:GAN-MC} illustrates how to use GAN-MC on option pricing.

\begin{algorithm}[ht]
\caption{GAN-MC for option pricing}
\label{alg:GAN-MC}
\KwIn{ $\mathbf{s}_{1,n},r,X,T_0,T, N_1,N_2,\alpha$ }
\KwOut{Estimated call or put option price $\widehat{C}$ or $\widehat{P}$}
\For{$d=1,2,\cdots,T$}{
Partition the stock price data $\mathbf{s}_{1,n}$ into a training set $\mathcal{S}^n_{d,T}$ according to~\eqref{equ:seperation}\;
Train GAN on training set $\mathcal{S}^n_{d,T}$ and check training loss\;
\uIf{loss does not collapse and $\left|\mathcal{S}^n_{d,T}\right|\geq N_1$}{\algorithmicbreak}
}
\For{$i=1,2,\cdots,N_2$}{
Generate random noise $\mathbf{Z}_i$ and denote $\tilde{\mathbf{S}}^{(i)}_T = G(\mathbf{Z}_i)=(\tilde{s}^{(i)}_{n+1},\tilde{s}^{(i)}_{n+2},\cdots,\tilde{s}^{(i)}_{n+T})$ \;
Calculate $\textrm{tSim}\left(\tilde{\mathbf{S}}^{(i)}_T, \mathbf{s}_{n-T,T} \right) $ between two stock prices by~\eqref{equ:tsim} for each $i$\; 
}
Sort the list of similarities to form $\pi_{\alpha}=(p_{1},p_{2},\cdots,p_{N_2})$ such that $\textrm{tSim}\left(\tilde{\mathbf{S}}^{(p_i)}_T, \mathbf{s}_{n-T,T} \right)\leq \textrm{tSim}\left(\tilde{\mathbf{S}}^{(p_j)}_T, \mathbf{s}_{n-T,T} \right)$ for every $1\leq i\leq j\leq N_2$, take $\mathcal{I}_{\alpha} = \{p_{i}:i\geq \lceil \alpha N_{2}\rceil \}$ \;
Calculate $\widehat{C}$ or $\widehat{P}$ accordingly\;
\Return{$\widehat{C}$ or $\widehat{P}$}
\end{algorithm}
\par First if we have the historical stock price data $\mathbf{s}_{1,n}$, we need to separate the data into realizations of $\mathbf{S}_T$ to construct training set for the following missions. The training set $\mathcal{S}^n_{d,T}$ with the length of sliding window $d$ and given $T$ is defined as
\begin{equation}\label{equ:seperation}
\mathcal{S}^n_{d,T}=\left\{ \mathbf{s}_{1,T},\mathbf{s}_{1+d,T},\cdots,\mathbf{s}_{1+\lfloor\frac{n-T}{d}\rfloor d,T }
\right\}
\end{equation}
Obviously, when $d$ equals to $0$, all the realizations will be the same and the training process of GAN would fail and collapse. With the increase of $d$, the overlap proportion between different training sample will become smaller, which will make it much more harder for generator to deceive the discriminator. But the size of training set would become smaller at the same time. Such trade-off is considered during the iteration, we start from a small $d$ and keep checking the training loss of GAN, if the loss does not collapse and the training set is not so small, we keep the GAN model for the following generation.
\par Inspired from the work of GAN~\cite{goodfellow2014generative}, the 1-D optimization process of GAN here could be stated as
\begin{equation*}
\min_{G}\max_{D}\mathbb{E}_{\mathbf{s}_{t,T}\sim p(\mathbf{S}_T)}\left[\log D(\mathbf{s}_{t,T})\right] + 
\mathbb{E}_{\mathbf{Z}\sim p(\mathbf{Z})}\left[\log(1-D(G(\mathbf{Z}))) \right]
\end{equation*}where $p(\mathbf{S}_T)$ is the probability distribution of $\mathbf{S}_T$ and $p(\mathbf{Z})$ is the probability distribution of noise $\mathbf{Z}$. We design both the generator and discriminator by fully connected networks(FCNNs)~\cite{rosenblatt1961principles}. Instead of convolutional neural network, the dense layer could better capture the structure information on 1 dimensional data.
\par After training, the generated data $\{\tilde{\mathbf{S}}^{(i)}_T\}_{i=1}^{N_2}$ could be treated as fake stock prices which follow the distribution of $\mathbf{S}_T$. Under Assumption~\ref{assump:date-indep}, these fake stock prices could be treated as the predictions for the following $T$ days, which means we could denote $\tilde{\mathbf{S}}^{(i)}_T =(\tilde{s}^{(i)}_{n+1},\tilde{s}^{(i)}_{n+2},\cdots,\tilde{s}^{(i)}_{n+T})$. However, there will be endogenous variation within real world stock prices, which makes Assumption~\ref{assump:date-indep} hard to be completely satisfied. Obviously the price of a option should rely much more on recent stock prices, and the old stock price will contribute less to the option pricing. Similar to the work of Cassisi~\cite{cassisi2012similarity}, here we define the similarity of two time series $X=(x_1,x_2,\cdots,x_n),Y=(y_1,y_2,\cdots,y_n)$ by 
\begin{equation}\label{equ:tsim}
\textrm{tSim}\left(X,Y\right) = \frac{1}{n}\sum_{i=1}^n\left(1 - \frac{|x_i - y_i|}{|x_i|+|y_i|}
\right)
\end{equation}
and we calculate $\textrm{tSim}\left(\tilde{\mathbf{S}}^{(i)}_T, \mathbf{s}_{n-T,T} \right) $ for each $i$. Some of the generated fake stock prices will be similar to recent stock trend, while others may look completely different and are close to earlier stock data. Thus we rank the list of similarities to form a unique order $\pi_\alpha=(p_1,p_2,\cdots,p_{N_2})$ such that the similarities are placed in a increasing trend, for every $1\leq i \leq j\leq N_2$ we have $\textrm{tSim}(\tilde{\mathbf{S}}^{(p_i)}_T, \mathbf{s}_{n-T,T} )\leq \textrm{tSim}(\tilde{\mathbf{S}}^{(p_j)}_T, \mathbf{s}_{n-T,T} )$. For a proportion parameter $\alpha\in(0,1)$ we take the generated fake stock prices which are close to recent market trend $\mathcal{I}_\alpha = \{p_{i}:i\geq \lceil \alpha N_{2}\rceil \}$ as the candidate samples index set for Monte Carlo. 
\par The basic theory of option pricing relies on risk neutral valuation. According to the original Monte Carlo methods used on European option pricing~\cite{boyle1977options}, the contract holder can purchase the stock at a future date $\frac{T_0}{\Delta t}$ at a price $X$ agreed upon in the contract. The payoff function of a call option could be stated as $f(S)=\max(S-X,0)$ where $S$ is the stock price at expiration date. We need the investment payoff is equal to the compound total return obtained by investing the option premium $C$, for European call option
\begin{equation*}
    \frac{1}{|\mathcal{I}_\alpha|}\sum_{i\in |\mathcal{I}_\alpha|}f(\tilde{s}^{(i)}_{n+\frac{T_0}{\Delta t}})=(1+r\Delta t)^{\frac{T_0}{\Delta t}}C
\end{equation*}
and solving the equation we get the GAN-based Monte Carlo estimation for call option
\begin{equation}\label{eq:call_option}
    \widehat{C}=\left(1+r\Delta t \right)^{-\frac{T_0}{\Delta t}}\frac{1}{|\mathcal{I}_{\alpha}|}\sum_{i\in \mathcal{I}_{\alpha}}\max\left( \tilde{s}^{(i)}_{n+\frac{T_0}{\Delta t}} - X,0\right)
\end{equation}Similarly, the only difference for put option pricing lies on the payoff function. The put option is a contract giving the option buyer the right to sell a specified amount of an underlying security, therefore the payoff function for put option should be $f(S)=\max(X-S,0)$. And we get the GAN-based Monte Carlo estimation for European put option
\begin{equation}\label{eq:put_option}
    \widehat{P}=\left(1+r\Delta t \right)^{-\frac{T_0}{\Delta t}}\frac{1}{|\mathcal{I}_{\alpha}|}\sum_{i\in \mathcal{I}_{\alpha}}\max\left( X - \tilde{s}^{(i)}_{n+\frac{T_0}{\Delta t}} ,0\right)
\end{equation}
\par For American option, since the contract allows holders to exercise their right at any time before and including the expiration date, the equation of the call option should be changed to
\begin{equation*}
    C\leq \frac{1}{|\mathcal{I}_\alpha|}\sum_{i\in |\mathcal{I}_\alpha|}f(\tilde{s}^{(i)}_{n+\frac{T_0}{\Delta t}})\leq (1+r\Delta t)^{\frac{T_0}{\Delta t}}C
\end{equation*} and we get the lower and upper bound for American call option
\begin{equation*}
    \left(1+r\Delta t \right)^{-\frac{T_0}{\Delta t}}\frac{1}{|\mathcal{I}_{\alpha}|}\sum_{i\in \mathcal{I}_{\alpha}}\max\left( \tilde{s}^{(i)}_{n+\frac{T_0}{\Delta t}} - X,0\right)\leq  \widehat{C}\leq \frac{1}{|\mathcal{I}_{\alpha}|}\sum_{i\in \mathcal{I}_{\alpha}}\max\left( \tilde{s}^{(i)}_{n+\frac{T_0}{\Delta t}} - X,0\right)
\end{equation*}Similarly for American put option
\begin{equation*}
    \left(1+r\Delta t \right)^{-\frac{T_0}{\Delta t}}\frac{1}{|\mathcal{I}_{\alpha}|}\sum_{i\in \mathcal{I}_{\alpha}}\max\left( X - \tilde{s}^{(i)}_{n+\frac{T_0}{\Delta t}} ,0\right)\leq \widehat{P} \leq \frac{1}{|\mathcal{I}_{\alpha}|}\sum_{i\in \mathcal{I}_{\alpha}}\max\left( \tilde{s}^{(i)}_{n+\frac{T_0}{\Delta t}} - X,0\right)
\end{equation*}
\par After getting the lower and upper bound for American put option, we could take the average of the two bounds as the final pricing formula for American option.
\par One significant advantage for Monte Carlo estimation methods is that the variance of statistics will decrease with the increase on sample size. Lower variance is associated with lower risk for investors. Such advantage could be stated as the following theorem.

\begin{theorem}
Given $r, T_0, \alpha, X$, $\textrm{Var}(\widehat{C})$ or $\textrm{Var}(\widehat{P})$ would not increase with the increase of $N_2$.
\end{theorem}
\begin{proof}
We design generator as fully connected network consisting of dense layers and activation layers, thus $G$ is a continuous function. And $\mathbf{Z}$ is a random noise, we generate $\{\mathbf{Z}_1,\mathbf{Z}_2,\cdots,\mathbf{Z}_{N_2}\}$ as independent and identically distributed random variables, if we denote $G(\mathbf{Z}_i)_j$ as the j-th element of $G(\mathbf{Z}_i)$, then the random variables $\max(G(\mathbf{Z}_i)_{n+\frac{T_0}{\Delta t}}-X,0) $ or $\max(X-G(\mathbf{Z}_i)_{n+\frac{T_0}{\Delta t}},0) $ are continuous functions of independent variables $\{\mathbf{Z}_i\}_{i=1}^{N_2}$. Therefore $\{\max(G(\mathbf{Z}_i)_{n+\frac{T_0}{\Delta t}}-X,0)\}_{i=1}^{N_2} $ are i.i.d. random variables and $\{ \max(X-G(\mathbf{Z}_i)_{n+\frac{T_0}{\Delta t}},0) \}_{i=1}^{N_2}$ are i.i.d. random variables. We know that $\mathcal{I}_\alpha\subseteq [N_2]$ and we take $\sigma^2_1 \coloneqq \textrm{Var}(\max(G(\mathbf{Z}_i)_{n+\frac{T_0}{\Delta t}}-X,0))$ and $\sigma^2_2 \coloneqq \textrm{Var}(\max(X-G(\mathbf{Z}_i)_{n+\frac{T_0}{\Delta t}},0) )$, then the variance of estimated call and put option could be stated as
\begin{equation*}
    \textrm{Var}(\widehat{C})\propto \frac{\sigma^2_1}{|\mathcal{I}_\alpha|} \quad \textrm{Var}(\widehat{P})\propto \frac{\sigma^2_2}{|\mathcal{I}_\alpha|}
\end{equation*}
therefore with the increase of $N_2$, the size of set $\mathcal{I}_{\alpha}$ will increase or stay the same. So if we keep $r, T_0, \alpha, X$ unchanged, $\textrm{Var}(\widehat{C})$ or $\textrm{Var}(\widehat{P})$ will decrease or stay the same.
\end{proof}

\subsection{GAN-MC for Forward and Futures Pricing}
A forward contract is a customized contract between two parties to buy or sell an asset at a specified price on a future date, and a futures contract is a standardized legal contract to buy or sell the underlying assets at a predetermined price for delivery at a specified time in the future. Forward contract is similar with futures, but settlement of forward contract takes place at the end of the contract, different with futures which settles on a daily basis. The underlying asset transacted is usually a commodity or financial instrument. Based on different commodities, securities, currencies or intangibles such as interest rates and stock indexes, the forward or futures could be categorized into markets like foreign exchange market, bond market, equity market and commodity market. Here we focus our pricing model on equity market and commodity market. And the valuation of equity forward or futures origins from a single stock, a customized basket of stocks or on an index of stocks, the valuation of commodity forward or futures depends on the cost of carry during the interim before delivery.
\subsubsection{Equity Market}
\par Similar to the setup in section~\ref{sec:option}, we denote the underlying stock or index price as $\mathbf{s}_{1,n}=(s_{1},s_{2},\cdots,s_{n})$, the value of equity forward or futures as $F^{\textrm{eq}}$, the historical data set of annual dividend per share of this stock as $\{D(t)\}_{t=1}^n$, the proportion of a year
before the delivery date as $T_0$ and $\Delta t$ is the time unit. Other parameters $r,T,N_1,N_2,\alpha$ are the same defined as section~\ref{sec:option}. We introduce algorithm~\ref{alg:GAN-MC future} to use GAN-MC on equity forward or futures pricing.

\begin{algorithm}[ht]
\caption{GAN-MC for equity futures/forward pricing}
\label{alg:GAN-MC future}
\KwIn{ $\mathbf{s}_{1,n},r,\{D(t)\}_{t=1}^n,T_0,T, N_1,N_2,\alpha$ }
\KwOut{Estimated equity futures or forward price $\widehat{F}^{\textrm{eq}}$}
\For{$d=1,2,\cdots,T$}{
Partition the stock price data $\mathbf{s}_{1,n}$ into a training set $\mathcal{S}^n_{d,T}$ according to~\eqref{equ:seperation}\;
Train GAN on training set $\mathcal{S}^n_{d,T}$ and check training loss\;
\uIf{loss does not collapse and $\left|\mathcal{S}^n_{d,T}\right|\geq N_1$}{\algorithmicbreak}
}
\For{$i=1,2,\cdots,N_2$}{
Generate random noise $\mathbf{Z}_i$ and denote $\tilde{\mathbf{S}}^{(i)}_T = G(\mathbf{Z}_i)=(\tilde{s}^{(i)}_{n+1},\tilde{s}^{(i)}_{n+2},\cdots,\tilde{s}^{(i)}_{n+T})$ \;
Calculate $\textrm{tSim}\left(\tilde{\mathbf{S}}^{(i)}_T, \mathbf{s}_{n-T,T} \right) $ between two stock or index prices by~\eqref{equ:tsim} for each $i$\; 
}
Sort the list of similarities to form $\pi_{\alpha}=(p_{1},p_{2},\cdots,p_{N_2})$ such that $\textrm{tSim}\left(\tilde{\mathbf{S}}^{(p_i)}_T, \mathbf{s}_{n-T,T} \right)\leq \textrm{tSim}\left(\tilde{\mathbf{S}}^{(p_j)}_T, \mathbf{s}_{n-T,T} \right)$ for every $1\leq i\leq j\leq N_2$, take $\mathcal{I}_{\alpha} = \{p_{i}:i\geq \lceil \alpha N_{2}\rceil \}$ \;
Fit a linear model $D(t)=at+b+\epsilon_t$ on set $\{D(t)\}_{t=1}^n$ where $\{\epsilon_t\}_{t=1}^n$ is the set of random noise and predict the annual dividend per share by $\widehat{D}(n+\frac{T_0}{\Delta t})=\hat{a}(n+\frac{T_0}{\Delta t})+\hat{b}$\;
Calculate 
\begin{equation}\label{eq:future_formula}
\widehat{F}^{\textrm{eq}}=s_n \cdot \exp{\left\{ \left(r-\frac{1}{|\mathcal{I}_\alpha|}\sum_{i\in\mathcal{I}_\alpha}\frac{\widehat{D}(n+\frac{T_0}{\Delta t})}{\tilde{s}^{(i)}_{n+\frac{T_0}{\Delta t}}} \right)T_0 \right\}}
\end{equation}
\Return{$\widehat{F}^{\textrm{eq}}$}
\end{algorithm}
\par The first step for equity market pricing is the same with option pricing. We need to partition the stock or index data into training set $\mathbf{S}^n_{d,T}$ with a proper sliding window $d$, which could make the GAN trained successfully with the given $T$. Similarly the stock or index price at different time $t$ will contribute differently for forward or futures pricing. We still use the rank of similarities between generated fake stock or index data $\{\tilde{\mathbf{S}}^{(i)}_T\}_{i=1}^{N_2}$ and $\mathbf{s}_{n-T,T}$ to control the effects of training sample at different time. Equity forward or futures prices are usually quoted in the same way as equity prices quoted in the underlying cash market by exchanges. And a pricing model is mainly used to calculate risk for a future contract, although it is utilized for computing both price and risk for a forward. The theoretical value of a equity forward or futures depends on the dividend model assumption~\cite{quail2009financial}, under dividend yield assumption, the theoretical equity forward's or futures' price is given by 
\begin{equation}\label{eq:futures-price}
    F^{\textrm{eq}}_{\tau} = s_t \exp{\left\{ \left(r-\frac{D(\tau)}{s_{\tau}}\right)(\tau - t) \right\}} 
\end{equation}Where $F^{\textrm{eq}}_{\tau}$ denotes the forward or futures price at delivery date $\tau$, $s_t$ and $s_\tau$ denote the stock or index price at time $\tau,t$ and $D(\tau)$ means the annual dividend per share at time $\tau$.
\par Given the historical annual dividend per share data $\{D(t)\}_{t=1}^{n}$, we need to predict the annual dividend per share at time $n+\frac{T_0}{\Delta t}$ which is the delivery date for the equity forward or futures. There are lots of methods for prediction and here we just consider the simple linear model
\begin{equation*}
    D(t)=at+b+\epsilon_t \quad \epsilon_t \overset{\text{i.i.d.}}{\sim}\mathcal{N}(0,\sigma^2) \quad t\in[n]
\end{equation*}Where $\epsilon_t$ denotes random noise and we use least squares estimation to get $\hat{a}=\frac{\sum_{t=1}^n(t-\bar{t})(D(t)-\overline{D(t)})}{\sum_{t=1}^n(t-\bar{t})^2}$ and $\hat{b}=\overline{D(t)} - \hat{a}\bar{t}$ where $\bar{t}=\frac{1}{n}\sum_{t=1}^n t$ and $\overline{D(t)}=\frac{1}{n}\sum_{t=1}^n D(t)$. Therefore we predict the annual dividend per share at the delivery date as $\widehat{D}(n+\frac{T_0}{\Delta t})=\hat{a}(n+\frac{T_0}{\Delta t})+\hat{b}$. 
\par Given the annual dividend per share, the dividend yield estimation by GAN-based Monte Carlo should be $\frac{1}{|\mathcal{I}_\alpha|}\sum_{i\in\mathcal{I}_\alpha}\frac{\widehat{D}(n+\frac{T_0}{\Delta t})}{\tilde{s}^{(i)}_{n+\frac{T_0}{\Delta t}}}$. Put it into the equity futures pricing formula~\ref{eq:futures-price} we get the GAN-based Monte Carlo estimation formula~\ref{eq:future_formula}.
\par Similar to the variance reduction in our method on option pricing. We could still reduce the variance of estimated prices. 
\begin{theorem}
Given $r, T_0, \alpha, \{D(t)\}_{t=1}^n$, $\textrm{Var}(\widehat{F}^{\textrm{eq}})$ would not increase with the increase of $N_2$.
\end{theorem}
\begin{proof}
We design generator as fully connected network consisting of dense layers and activation layers, thus $G$ is a continuous function. And $\mathbf{Z}$ is a random noise, we generate $\{\mathbf{Z}_1,\mathbf{Z}_2,\cdots,\mathbf{Z}_{N_2}\}$ as independent and identically distributed random variables, thus random variables set $\{X_{i}=\frac{\widehat{D}(n+\frac{T_0}{\Delta t})}{G(\mathbf{Z}_{i})_{n+\frac{T_0}{\Delta t}}} \}_{i\in\mathcal{I}_\alpha}$ are collections of independent and identically distributed random variables if $\{D(t)\}_{t=1}^n$ is fixed because $X_i$ is continuous functions of independent variables $\{\mathbf{Z}_i\}_{i=1}^{N_2}$. We denote $\mu_0$ as the mean of $X_{i}$ and $\sigma^2_0$ as the variance of $X_{i}$, then
\begin{equation*}
\mathbb{E}\left(\frac{1}{|\mathcal{I}_\alpha|}\sum_{i\in\mathcal{I}_{\alpha}}X_i\right)=\mu_0 \quad \textrm{Var}\left(\frac{1}{|\mathcal{I}_\alpha|}\sum_{i\in\mathcal{I}_\alpha}X_i \right)=\frac{1}{|\mathcal{I}_\alpha|}\sigma^2_0
\end{equation*}And with the increase of $N_2$, $\mathcal{I}_\alpha \subseteq [N_2]$, $|\mathcal{I}_\alpha|$ will increase or stay the same, and $\text{Var}(\frac{1}{|\mathcal{I}_\alpha|}\sum_{i\in\mathcal{I}_\alpha}X_i)$ will not increase. If $\mathbb{E}(\widehat{F}^{\textrm{eq}})$ keeps unchanged, $\text{Var}(\widehat{F}^{\textrm{eq}})$ will decrease or stay the same.
\end{proof}
\subsubsection{Commodity Market}
For commodity forward contract or futures, the underlying asset could usually be divided into food, energy and materials. Similar to the parameter setup in section~\ref{sec:option}, we denote the underlying commodity spot price as $\mathbf{s}_{1,n}=(s_{1},s_{2},\cdots,s_{n})$, the average cost of carry from time $t$ to $\tau$ as $P(t,\tau)$, the historical commodity forward contract or futures price as $\{F^{\textrm{co}}_t\}_{t=n-N_3}^{n}$, the proportion of a year before the delivery date as $T_0$ and $\Delta t$ is the time unit. Other parameters $r,T,N_1,N_2,\alpha$ are the same defined as section~\ref{sec:option}. We introduce algorithm~\ref{alg:GAN-MC co} to use GAN-MC on commodity forward or futures pricing.

\begin{algorithm}[ht]
\caption{GAN-MC for commodity futures/forward pricing}
\label{alg:GAN-MC co}
\KwIn{ $\mathbf{s}_{1,n},r,\{F^{\textrm{co}}_t\}_{t=n-N_3}^{n},T_0,T, N_1,N_2,N_3,\alpha$ }
\KwOut{Estimated commodity forward or futures price $\widehat{F}^{\textrm{co}}$}
\For{$d=1,2,\cdots,T$}{
Partition the spot price data $\mathbf{s}_{1,n}$ into a training set $\mathcal{S}^n_{d,T}$ according to~\eqref{equ:seperation}\;
Train GAN on training set $\mathcal{S}^n_{d,T}$ and check training loss\;
\uIf{loss does not collapse and $\left|\mathcal{S}^n_{d,T}\right|\geq N_1$}{\algorithmicbreak}
}
\For{$i=1,2,\cdots,N_2$}{
Generate random noise $\mathbf{Z}_i$ and denote $\tilde{\mathbf{S}}^{(i)}_T = G(\mathbf{Z}_i)=(\tilde{s}^{(i)}_{n+1},\tilde{s}^{(i)}_{n+2},\cdots,\tilde{s}^{(i)}_{n+T})$ \;
Calculate $\textrm{tSim}\left(\tilde{\mathbf{S}}^{(i)}_T, \mathbf{s}_{n-T,T} \right) $ between two spot prices by~\eqref{equ:tsim} for each $i$\; 
}
Sort the list of similarities to form $\pi_{\alpha}=(p_{1},p_{2},\cdots,p_{N_2})$ such that $\textrm{tSim}\left(\tilde{\mathbf{S}}^{(p_i)}_T, \mathbf{s}_{n-T,T} \right)\leq \textrm{tSim}\left(\tilde{\mathbf{S}}^{(p_j)}_T, \mathbf{s}_{n-T,T} \right)$ for every $1\leq i\leq j\leq N_2$, take $\mathcal{I}_{\alpha} = \{p_{i}:i\geq \lceil \alpha N_{2}\rceil \}$ \;
Estimate cost of carry $\widehat{P}\left(n,n+\frac{T_0}{\Delta t}\right)$ by formula~\eqref{eq:co_p}\;
Calculate $\widehat{F}^{\textrm{co}}$ from formula ~\eqref{eq:co_price}\;
\Return{$\widehat{F}^{\textrm{co}}$}
\end{algorithm}

\par Similar to the equity pricing formula~\eqref{eq:futures-price}, the theoretical commodity forward price is based on its current spot price, plus the cost of carry during the interim before delivery~\cite{black1976pricing}. The simple commodity forward contract or futures could be expressed as
\begin{equation*}
    F_{\tau}^{\textrm{co}} = (s_t + P(t,\tau))\exp{\{r(\tau - t)\}}
\end{equation*}Where $F^{\textrm{co}}_{\tau}$ denotes the forward or futures price at delivery date $\tau$, $s_t$ denotes the commodity spot price at time $t$. Then if we use $s_{\tau}$ to replace the term $s_{t}\exp{r(\tau - t)}$, the price of commodity forward or futures would become $F_{\tau} = s_{\tau} + P(t,\tau)\exp{\{r(\tau - t)\}}$. Like the previous notation in algorithms, for pricing the commodity forward or futures at time $n+\frac{T_0}{\Delta t}$, we use empirical estimation of $P\left(n,n+\frac{T_0}{\Delta t}\right)$
\begin{equation}\label{eq:co_p}
\widehat{P}\left(n,n+\frac{T_0}{\Delta t}\right) = \frac{1}{N_3+1}\sum_{j=0}^{N_3}\left[\frac{F^{\textrm{co}}_{n-j}}{\exp{(rT_0)}} - s_{n-j}
\right]
\end{equation}Where $N_3+1$ is the sample size for empirical estimation. Then we use GAN-MC to estimate $\tilde{s}_{n+\frac{T_0}{\Delta t}}$. Analog to equity futures or forward pricing formula~\ref{eq:future_formula}, the commodity futures or forward pricing formula would be
\begin{equation}\label{eq:co_price}
    \widehat{F}^{\textrm{co}} = \frac{1}{|\mathcal{I}_\alpha|}\sum_{i\in \mathcal{I}_\alpha}\tilde{s}^{(i)}_{n+\frac{T_0}{\Delta t}}+\widehat{P}\left(n,n+\frac{T_0}{\Delta t}\right)\exp{(rT_{0})}
\end{equation}
\par Similar to the variance reduction theorem in equity forward or futures pricing, we claim
\begin{theorem}
Given $r, T_0, \alpha, \{F^{\textrm{co}}_t\}_{t=n-N_3}^{n}$ fixed, $\textrm{Var}(\widehat{F}^{\textrm{co}})$ would not increase with the increase of $N_2$.
\end{theorem}

\section{Experiments}
\subsection{Stock and Index Price Prediction}
The accuracy of Monte Carlo pricing models is highly based on the accuracy and variation of prediction on underlying assets. Equation~\eqref{eq:call_option}\eqref{eq:put_option}\eqref{eq:future_formula} and\eqref{eq:co_price} have shown the pricing results are directly correlated to the generated stock prices in the future. Thus before we test our pricing models on real-world market data, we first check the generated fake stock prices tracks after GAN training.
\begin{figure}[htbp]
  \centering
  \includegraphics[width=1.0\textwidth]{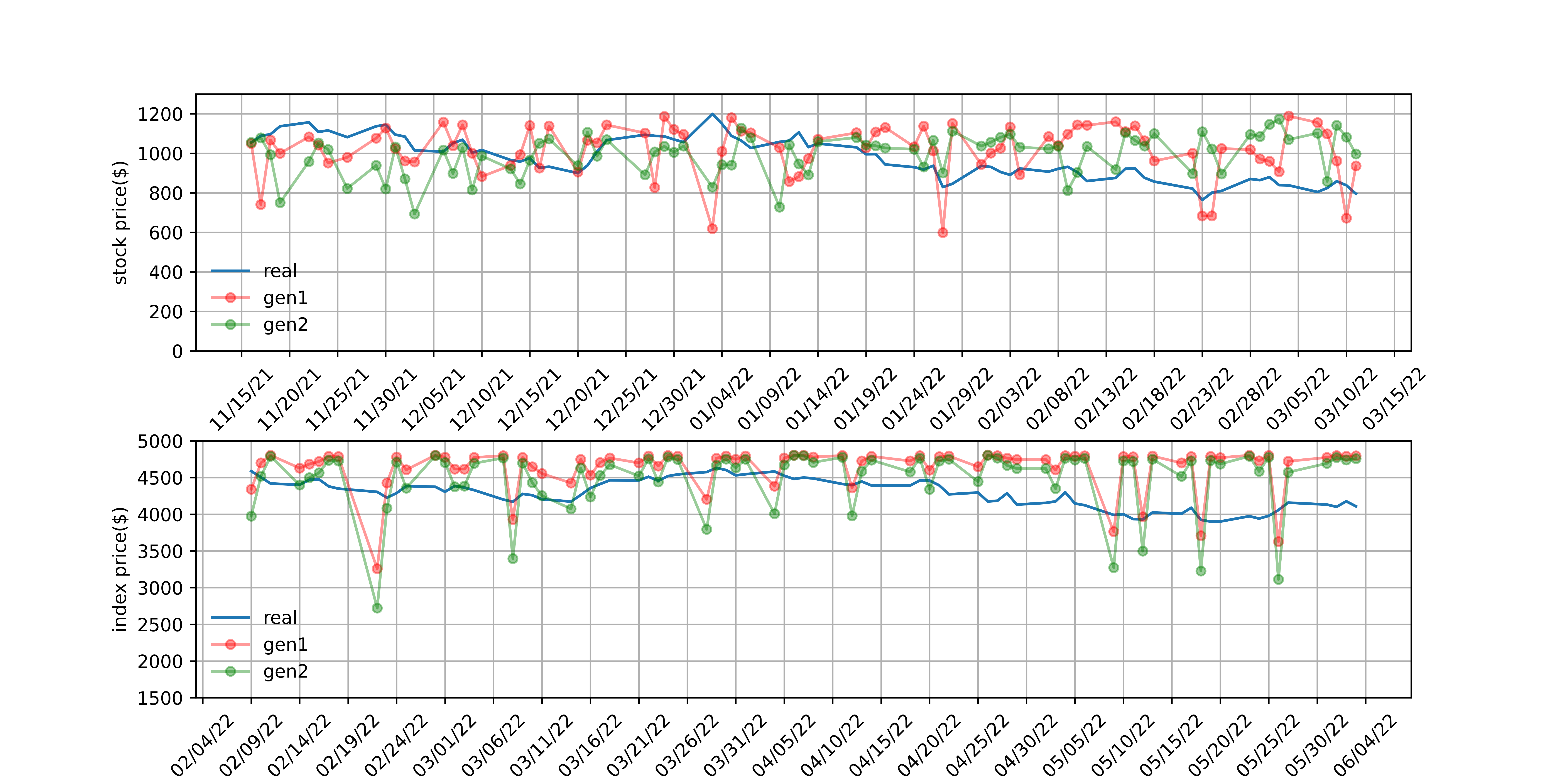}
  \caption{Stock and index prediction for the following 80 days. The upper figure shows prediction results for TSLA, the blue line denotes the real stock price of TSLA from 11/16/2021 to 3/11/2022, the red dotted line and green dotted line denote two generated samples from generator $G$. The lower figure shows prediction results for index S\&P 500, the blue line denotes real index price from 2/9/2022 to 6/3/2022, the red dotted line and green dotted line indicate two generated samples from generator $G$.}
  \label{fig:pred}
\end{figure}
\par We collect TSLA historical daily stock prices from 3/22/2019 to 4/11/2022 and S\&P 500 index prices from 6/14/2019 to 7/6/2022 for training and prediction. Excluding all the holidays and weekends, there are 771 daily data in total. We set $n$ in $\mathbf{s}_{1,n}$ to 670, the dimension of generator output $T$ to 128, the sample size threshold $N_{1}$ to 270, the size of generation $N_2$ to $1600$ and the proportion parameter $\alpha$ to $0.8$. All the financial data is collected from Bloomberg database.
\par We follow the procedures of fake price generation in algorithm~\ref{alg:GAN-MC} and algorithm~\ref{alg:GAN-MC future}. Then we form the index set $\mathcal{I}_\alpha$ and pick two generated price tracks as prediction in comparison to real market data. As shown in fig~\ref{fig:pred}, the red dotted line and green dotted line indicate two generated samples $G(\mathbf{Z}_i),G(\mathbf{Z}_j)$ where $i,j \in\mathcal{I}_\alpha$. Overall the generated fake stock tracks are more turbulent than real stock prices, there are sharp decreases and sharp increases between two single day. In TSLA stock price prediction, the generated samples distribute around real stock prices. And as for pricing, the variation between different generations would make Monte Carlo estimation more accurate. An interesting discovery in S\&P 500 index price prediction is that two generated data share a same trend, which means they have a large linear correlation. The trends are similar but the predicted index prices are not exactly the same for each day. The prediction experiments show that our GAN could generate fitful samples for pricing.

\subsection{Option Pricing}
Before showing our experiments on option pricing, we first introduce other basic models for option pricing.
\subsubsection{Black–Scholes}
The Black-Scholes model, also known as the Black-Scholes-Merton (BSM) model estimates the theoretical value of derivatives based on other investment instruments, taking into account the impact of time and other risk factors. Raised by Fisher and Myron~\cite{black2019pricing}, the Black-Scholes model for option pricing assumes no dividends are paid out during the life of the option, markets are random, there are no transaction costs in buying the option, the risk-free rate and volatility of the underlying asset are known and constant, the returns of the underlying asset are normally distributed and the option can only be exercised at expiration. Then if we denote $C$ as the call option price, $N(x)$ denotes the standard normal cumulative distribution function, $N(x)=\frac{1}{\sqrt{2\pi}}\int_{-\infty}^{x}e^{-z^2/2}dz$. And we denote $X$ as the strike price, $s_{t}$ as the spot price at time $t$, $T_1$ as the date of option expiration, $r$ as annual risk-free interest rate, $\sigma$ as the standard deviation of the stock's returns, which is known as volatility. Then the formula for call option
\begin{equation*}
    C  = N(d_1)s_t - N(d_2)Xe^{-r(T_1-t)}
\end{equation*}where
\begin{align*}
    d_1 & = \frac{1}{\sigma \sqrt{T_1-t}}\left[\log\left(\frac{s_t}{X}\right) +\left(r+\frac{\sigma^2}{2}\right)(T_1-t) \right]\\
    d_2 & = d_1 - \sigma\sqrt{T_1- t}
\end{align*}And the formula for put option
\begin{equation*}
    P = N(-d_2)Xe^{-r(T_1-t)} - N(-d_1)s_t
\end{equation*}where $d_1,d_2$ are same defined in call option.
\par In implementing Black-Scholes formula, we actually use the pricing tools in Bloomberg where volatility is automatically calculated according to market data.

\subsubsection{Radial Basis Function Network}
As mentioned in the work by Hutchinson and Poggio~\cite{hutchinson1994nonparametric}, the Radial Basis Function(RBF) network could be used for data fitting\cite{poggio1990networks}. If we replace the Euclidean norm with matrix norm, the fitting process could be expressed as
\begin{equation*}
     \hat{f}(x) = \sum_{i=1}^{k}c_{i}h_{i}((x-z_{i})^{\top}W^{\top}W(x-z_{i}))+\alpha_{0} + \alpha_{1}^{\top}x
\end{equation*}where $W,c_i,z_i,\alpha_0,\alpha_1$ are parameters to be optimized on, and $h_i$ is the basis function, which could be either Gaussian or multiquadric. In pricing model we add one more sigmoid layer as output. The augmented network will be of the form $g(\hat{f}(x))$ where $g(u)=1/(1+e^{-u})$.
\par In implementing RBF network, we set the basis function as multiquadric and set the input for the model as $\hat{g}(s/X,1,T_1-t)$ where $s$ is stock price, $X$ is strike price and $T_1 - t$ is the time until maturity.

\subsubsection{Multilayer Perceptrons Regression}
Still in the paper by Hutchinson and Poggio~\cite{hutchinson1994nonparametric}, Multilayer Perceptrons(MLPs)~\cite{rumelhart1985learning} are classical methods for high dimension regression. Consisting of fully connected networks, a general formulation of MLPs with univariate output could be written as
\begin{equation*}
    \hat{f}(x) =h\left(\sum_{i=1}^{k}\delta_{i}h(\beta_{0i}+\beta_{1i}^{\top}x)+\delta_{0}\right)
\end{equation*}Where $h(\cdot)$ is the sigmoid function, and $\delta_i,\delta_0,\beta_{0i},\beta_{1i}$ are parameters to be optimized on. Unlike the RBF network, the nonlinear function $h$ in MLP is usually fixed for the entire network.
\par In implementing MLP regression, we set the input of the model as $\hat{f}(s/X,1,T_1-t)$ where $s$ is stock price, $X$ is strike price and $T_1 - t$ is the time until maturity.

\subsubsection{Projection Pursuit Regression}
Still mentioned in the work of Hutchinson and Poggio~\cite{hutchinson1994nonparametric}, projection pursuit regression(PPR)~\cite{friedman1981projection} was developed for high-dimensional regression. PPR models are composed of projections of the data and estimating nonlinear combining functions from data. The regression model could be stated as
\begin{equation*}
    \hat{f}(x) = \sum_{i=1}^{k}\delta_{i}h_{i}(\beta_i^\top x) +\delta_{0}
\end{equation*}where $h_i$ are functions estimated from data, $\delta_i,\delta_0,\beta_i$ are parameters to be optimized on, and $k$ is the number of projections.
\par In implementing PPR, we set the input of the model as $\hat{f}(s/X,1,T_1-t)$ where $s$ is stock price, $X$ is strike price and $T_1 - t$ is the time until maturity.

\subsubsection{Monte Carlo}
In fact there are many methods of Monte Carlo estimation for option pricing, and the main difference lays on the assumptions of stock price generation. Here we adapt the methods from Kevin~\cite{brewer2012geometric}. We assume that the process of stock price follows geometric Brownian motion. If we denote $s$ as stock price, $t$ as time, then 
\begin{equation*}
    ds = \mu sdt + \sigma sdz \quad dz=\epsilon\sqrt{t}
\end{equation*}where $\mu$ is the drift parameter, $\sigma$ is volatility parameter, $\epsilon$ is a random draw from standard normal distribution. Given stock price at time $t$, we relate drift parameter to expected stock price and exercise price $\mu = \frac{1}{T_1 - t}\log(X/s_t)$ where $T_1$ is the date of option expiration. Following the stochastic process, we generate $N_2$ different stock price tracks $(s^{(i)}_t,\tilde{s}^{(i)}_{t+1},\cdots,\tilde{s}^{(i)}_{T_1})_{i=1}^{N_2}$ for pricing. Then if we denote $\Delta t$ as time unit, r as annual risk-free interest rate, the estimated price $\widehat{C}=(1+r\Delta t)^{T_1 - t}\frac{1}{N_2}\sum_{i=1}^{N_2}\max(\tilde{s}^{(i)}_{T_1}-X,0)$ and $\widehat{P}=(1+r\Delta t)^{T_1 - t}\frac{1}{N_2}\sum_{i=1}^{N_2}\max(X-\tilde{s}^{(i)}_{T_1},0)$ for European option pricing and take the average of lower and upper bounds for American option pricing.
\par In implementing Monte Carlo, we collect the volatility parameter for each option from Bloomberg.

\subsubsection{Experiments}
We first test our model on option pricing, which covers European and American options, call and put options. As for American options, we collect TSLA historical daily stock prices data from 3/22/2019 to 1/27/2022 for GAN training and Monte Carlo estimation, excluding all the holidays and weekends, there are 720 daily data in total. We collect TSLA call option data with strike price, last deal price and expiration date from 2/25/2022 to 3/10/2022, there are 8 different expiration dates for each day, and there are 10 different strike prices for each expiration date. We collect TSLA put option data from 3/21/2022 to 4/1/2022 and the other setup is same to call. Thus we have 800 option data for call and put, we use 720 data for training RBF network, MLP regression, PPR and Linear Regression, and the rest for testing.
\par We set the dimension of generator $T$ to 128, the sample size threshold $N_1$ to 290, the size of generation $N_2$ to 5120 and the proportion parameter $\alpha$ to 0.8. We collect the annual risk-free interest rate from Bloomberg. We use mean absolute percentage error(MAPE) as the metric for model evaluation, where
\begin{equation*}
    \text{MAPE}=\frac{100\%}{N_{\text{test}}}\sum_{j=1}^{N_{\text{test}}}\frac{|\widehat{V}-V|}{V}
\end{equation*}where $\widehat{V}$ is the predicted call or put option price, $V$ is the real option price and $N_{\text{test}}$ is the size of test set.

\begin{table}[htbp]
\centering
\begin{tabular}{ cc }
\begin{tabular}{|c|c|}
\hline
Model & MAPE\\
\hline
GAN-MC & 1.42\%\\
MC & 2.55\% \\
*BS & 1.33\%\\
RBF Network & 8.75\%\\
MLP Regression& 6.42\%\\
PPR & 4.52\%\\
LR & 10.3\%\\
LR-ITM & 10.88\%\\
LR-OTM & 10.55\%\\
\hline
\end{tabular} &
\begin{tabular}{|c|c|}
\hline
Model & MAPE\\
\hline
*GAN-MC & 1.02\%\\
MC & 1.91\% \\
BS & 2.82\%\\
RBF Network & 12.1\%\\
MLP Regression& 9.46\%\\
PPR & 10.69\%\\
LR & 15.32\%\\
LR-ITM & 12.59\%\\
LR-OTM & 17.25\%\\
\hline
\end{tabular}
\end{tabular}
\caption{\label{tab:table1}Performance table for TSLA option pricing. Here GAN-MC means our model, MC means Monte Carlo estimation, BS is the abbreviation of Black-Scholes model, LR is the linear model using all the option data, LR-ITM is In-the-Money Linear Model, LR-OTM is Out-of-the-Money Linear Model. The left table is the performance on call option and the right table collects performance for put option.}
\end{table}
\par As seen from table~\ref{tab:table1}, our model's performance is close to Black-Scholes model on call option pricing and our model reaches state-of-the-art for TSLA put option pricing. For a single price prediction, a smaller variance of $\widehat{C}$ or $\widehat{P}$ would make the prediction more accurate, we set $N_2$ as a relative large number to make the Monte Carlo estimation more accurate. And in both two cases our GAN-MC performs better than MC only model, which means GAN holds better generation capacity for real market data. As for the three non-parametric deep learning models, RBF network, MLP regression and PPR, the performances on TSLA call and put option are quite similar for the equivalence of three representations.
\par Apart from common stock prices, our model could still work on index prices. As for European options, we test our model on S\&P 500 index option price. Similar to TSLA, for call option pricing, we collect S\&P 500 historical daily index prices data from 6/14/2019 to 1/10/2022 for GAN training and Monte Carlo estimation, excluding all the holidays and weekends, there are 650 daily data in total. We collect S\&P 500 call option(SPXW) data with strike price, last deal price and expiration date from 4/6/2022 to 4/28/2022, there are 8 different expiration dates mixed with call option types for each day and 10 different strike prices for each expiration date. After removing the part of SPX, we have 700 S\&P 500 call option data. We use 650 data for training and the rest for testing. As for put option, we collect S\&P 500 historical daily index prices data from 1/6/2020 to 9/30/2022 for GAN training and Monte Carlo estimation, excluding all the holidays and weekends, there are 690 daily data in total. We collect S\&P 500 put option(SPXW) data with strike price, last deal price and expiration date from 10/7/2022 to 10/28/2022, we use 690 data for training and the rest for testing.

\begin{table}[htbp]
\centering
\begin{tabular}{ cc }
    \begin{tabular}{|c|c|}
\hline
Model & MAPE\\
\hline
*GAN-MC & 4.50\%\\
MC & 7.27\% \\
BS & 19.96\%\\
RBF Network & 17.00\%\\
MLP Regression& 14.2\%\\
PPR & 10.40\%\\
LR & 6.72\%\\
LR-ITM & 6.52\%\\
LR-OTM & 7.82\%\\
\hline
\end{tabular}&
\begin{tabular}{|c|c|}
\hline
Model & MAPE\\
\hline
*GAN-MC & 2.68\%\\
MC & 20.61\% \\
BS & 8.20\%\\
RBF Network & 16.83\%\\
MLP Regression& 12.28\%\\
PPR & 19.97\%\\
LR & 10.58\%\\
LR-ITM & 10.39\%\\
LR-OTM & 11.04\%\\
\hline
\end{tabular}
\end{tabular}
    \caption{\label{tab:table2}Performance table for S\&P 500 Weeklys(SPXW) options pricing. The left table is the performance on call option and the right table collects performance for put option.
}
\end{table}
\par As seen from table~\ref{tab:table2}, our model performs best among all the pricing models. Maybe sometimes the last deal price of SPXW will not lay on the range of bid price and ask price, Black-Scholes performs badly on SPXW pricing. And it seems the linear trend is significant in SPXW, therefore linear models perform better than other datasets.

\subsection{Equity Forward or Futures Pricing}\label{sec:eq}
Similar to the methods used in option pricing, we conduct experiments for equity futures pricing by GAN-MC, RBF network, MLP regression, PPR, linear regression and Monte Carlo. We first collect S\&P 500 historical daily index prices data from 6/14/2019 to 4/21/2022 for GAN training and Monte Carlo estimation, excluding all the holidays and weekends, there are 720 daily data in total. Then we collect E-Mini S\&P 500 futures data, which includes ESU22(delivery at 9/16/2022) and ESZ22(delivery at 12/16/2022) from 7/12/2021 to 7/6/2022. In addtion we collect historical E-Mini S\&P 500 futures ESH21(delivery at 3/18/2022) from 1/8/2021 to 3/18/2022. All the futures data includes the last futures price, remaining time before delivery date and S\&P 500 index price. We use 720 data for training RNF network, MLP regression, PPR and linear regression, and the rest for testing. Different from option pricing, we use stock price $s$ and the time until delivery $T_1-t$ as the input variables for non-parametric machine learning models for equity futures pricing.
\par As for Monte Carlo, following the assumption of geometric Brownian motion, we estimate drift parameter for lack of strike prices
\begin{equation*}
    \hat{\mu} = \frac{1}{n} \sum_{t=1}^{n}\frac{(ds)_{t}}{s_{t}dt}  
\end{equation*}where $(ds)_t = s_{t+1} - s_{t}$. And we estimate the volatility parameter as historical volatility
\begin{equation*}
    \hat{\sigma}=\sqrt{\frac{1}{n-1}\sum_{t=1}^{n}(R_{t} - \bar{R})^{2}}
\end{equation*}with $R_t = \log(s_t/s_{t-1}) $ and $\bar{R}$ is the mean of $R_t$. Then the Monte Carlo estimation of equity forward or futures price is given by $\widehat{F}^{\textrm{eq}}=s_t \cdot \exp{\left\{ \left(r-\frac{1}{N_2}\sum_{i=1}^{N_2}\frac{\widehat{D}(T_1)}{\tilde{s}^{(i)}_{T_1}} \right)\frac{T_1-t}{\Delta t} \right\}}$, where $t$ is current date, $T_1$ is the delivery date, $\Delta t$ is time unit and $\tilde{s}^{(i)}_{T_1}$ is the generated stock price in track $(s^{(i)}_{t},\tilde{s}^{(i)}_{t+1},\cdots,\tilde{s}^{(i)}_{T_1})_{i=1}^{N_2}$.

\par We set the dimension of generator $T$ to 128, the sample size threshold $N_1$ to 290, the size of generation $N_2$ to 5120 and the proportion parameter $\alpha$ to 0.8. We collect the annual risk-free interest rate from Bloomberg. And we use mean absolute percentage error(MAPE) as the metric for model evaluation.
\begin{table}[htbp]
\centering
    \begin{tabular}{|c|c|}
\hline
Model & MAPE\\
\hline
*GAN-MC & 0.03\%\\
MC & 0.36\% \\
RBF Network(Gauss) & 0.17\%\\
RBF Network(Sqrt) & 0.31\%\\
MLP Regression & 0.43\%\\
PPR & 0.11\%\\
LR & 0.31\%\\
\hline
\end{tabular}
    \caption{\label{tab:table3}Performance table for E-Mini S\&P 500 futures pricing. Here RBF Network(Gauss) means we use Gaussian as basis functions and RBF Network(Sqrt) means we use multiquadric as basis functions.
}
\end{table}
\par As seen from table~\ref{tab:table3} our model still performs best on equity futures pricing. For a single price prediction, a smaller variance of $\widehat{F}^{\textrm{eq}}$ would make the prediction value more accurate. We set $N_2$ as a relative large number to make the Monte Carlo estimation more accurate. And in all cases our GAN-MC performs better than MC only model, which means GAN holds better stability and generation capacity for market data.

\subsection{Commodity Forward or Futures Pricing}
Quite similar to the methods used in section~\ref{sec:eq}, we conduct experiments for commodity forward contract pricing by GAN-MC, RBF network, MLP regression, PPR, linear regression and Monte Carlo. We first collect LME copper spot daily price data from 12/4/19 to 9/2/22  for GAN training and Monte Carlo estimation, excluding all the holidays and weekends, there are 700 daily data in total. Then we collect LME copper 3 months rolling forward daily price data from 10/17/18 to 9/30/22. All the forward data includes the last forward price, remaining time before delivery date, which is three months and spot price. We use 700 data for training RNF network, MLP regression, PPR and linear regression, and the rest for testing. We use spot price $s$ and the time until delivery as the input variables for non-parametric machine learning models for commodity forward or futures pricing.
\par Similar to the method used in equity futures pricing, the Monte Carlo estimation of commodity forward contract or futures price is given by $\widehat{F}^{\textrm{co}}=\frac{1}{N_2}\sum_{i=1}^{N_2}\tilde{s}^{(i)}_{T_1}+\widehat{P}(t,T_1)\exp{(r(T_1-t))}$ where $t$ is current date, $T_1$ is the settlement date, and $\tilde{s}^{(i)}_{T_1}$ is the generated stock price in track $(s^{(i)}_{t},\tilde{s}^{(i)}_{t+1},\cdots,\tilde{s}^{(i)}_{T_1})_{i=1}^{N_2}$.
\par
We set the dimension of generator $T$ to 128, the sample size threshold $N_1$ to 290, the size of generation $N_2$ to 5120, sample size for estimating cost of carry $N_3$ to 50 and the proportion parameter $\alpha$ to 0.8. We collect the annual risk-free interest rate from Bloomberg. And we use mean absolute percentage error(MAPE) as the metric for model evaluation.

\begin{table}[htbp]
\centering
    \begin{tabular}{|c|c|}
\hline
Model & MAPE\\
\hline
*GAN-MC & 0.08\%\\
MC & 0.53\% \\
RBF Network(Gauss) & 0.90\%\\
RBF Network(Sqrt) & 2.33\%\\
MLP Regression & 1.11\%\\
PPR & 1.24\%\\
LR & 1.13\%\\
\hline
\end{tabular}
    \caption{\label{tab:table5}Performance table for LME copper 3 month forward pricing.
}
\end{table}
\par The results in table~\ref{tab:table5} show that our model performs best on commodity forward pricing. If we compare GAN-MC with MC, the better performance of our model proves the generative network's efficiency and capacity during generation. Apart from our model, Monte Carlo and RBF Network(Gauss) models perform greatly on LME copper forward pricing. 

\par Apart from commodity forward contract, GAN-MC could still handle the commodity futures cases. We then test our model on crude oil futures. Similar to copper, we collect Cushing, OK WTI crude oil historical daily spot price from 7/11/2019 to 8/30/2022 for GAN training and Monte Carlo estimation, excluding all the holidays and weekends, there are 700 daily data in total. Then we collect CLV2, which is the WTI crude oil future settled on October 2022 from 2/5/2019 to 2/9/2022, and CLX2, which is settled on November 2022 from 2/5/2020 to 2/9/2022. All the futures data includes the last futures price, remaining time before delivery date and WTI crude oil spot price. We use 700 daily data for training RNF network, MLP regression, PPR and linear regression, and the rest for testing.  
\begin{table}[htbp]
\centering
    \begin{tabular}{|c|c|}
\hline
Model & MAPE\\
\hline
*GAN-MC & 0.58\%\\
MC & 7.44\% \\
RBF Network(Gauss) & 2.59\%\\
RBF Network(Sqrt) & 1.88\%\\
MLP Regression & 4.75\%\\
PPR & 2.26\%\\
LR & 4.29\%\\
\hline
\end{tabular}
    \caption{\label{tab:table4}Performance table for WTI crude oil futures pricing.
}
\end{table}
\par The results in table~\ref{tab:table4} show that our model performs best on commodity futures pricing. Such success results from the capacity, generating ability and the variance reduction properties of GAN-MC. Apart from our model, RBF Network(Sqrt) and PPR models perform greatly on WTI crude oil futures pricing.
\section{Conclusion}
 All the success of our model on different real market derivatives pricing proves the correctness of our GAN-MC model. GAN is a powerful tool for capturing the trend and variation of the underlying asset prices like stock or index price. Monte Carlo could be used for reducing the variance of estimators given independent sequences and efficient for derivatives pricing. Although parametric derivatives pricing formulas are preferred when they are available, our result show that generative model-based Monte Carlo alternatives could be useful substitutes when  arbitrage-based pricing formula or non-parametric pricing model fails. While our results are promising, we can not claim our approach will be successful in general, we have not covered swap and other derivatives pricing yet and we hope to provide a more comprehensive analysis of these alternatives in the near future.

\newpage

\bibliography{refs}

\begin{thebibliography}{27}
\providecommand{\natexlab}[1]{#1}
\providecommand{\url}[1]{\texttt{#1}}
\expandafter\ifx\csname urlstyle\endcsname\relax
  \providecommand{\doi}[1]{doi: #1}\else
  \providecommand{\doi}{doi: \begingroup \urlstyle{rm}\Url}\fi

\bibitem[Black(1976)]{black1976pricing}
Fischer Black.
\newblock The pricing of commodity contracts.
\newblock \emph{Journal of financial economics}, 3\penalty0 (1-2):\penalty0
  167--179, 1976.

\bibitem[Merton(1973)]{merton1973theory}
Robert~C Merton.
\newblock Theory of rational option pricing.
\newblock \emph{The Bell Journal of economics and management science}, pages
  141--183, 1973.

\bibitem[Karoui et~al.(1998)Karoui, Jeanblanc-Picqu{\`e}, and
  Shreve]{karoui1998robustness}
Nicole~El Karoui, Monique Jeanblanc-Picqu{\`e}, and Steven~E Shreve.
\newblock Robustness of the black and scholes formula.
\newblock \emph{Mathematical finance}, 8\penalty0 (2):\penalty0 93--126, 1998.

\bibitem[Wu(2004)]{wu2004pricing}
Hsien-Chung Wu.
\newblock Pricing european options based on the fuzzy pattern of black--scholes
  formula.
\newblock \emph{Computers \& Operations Research}, 31\penalty0 (7):\penalty0
  1069--1081, 2004.

\bibitem[Magdziarz(2009)]{magdziarz2009black}
Marcin Magdziarz.
\newblock Black-scholes formula in subdiffusive regime.
\newblock \emph{Journal of Statistical Physics}, 136\penalty0 (3):\penalty0
  553--564, 2009.

\bibitem[Carmona and Durrleman(2005)]{carmona2005generalizing}
Ren{\'e} Carmona and Valdo Durrleman.
\newblock Generalizing the black-scholes formula to multivariate contingent
  claims.
\newblock \emph{Journal of computational finance}, 9\penalty0 (2):\penalty0 43,
  2005.

\bibitem[Hutchinson et~al.(1994)Hutchinson, Lo, and
  Poggio]{hutchinson1994nonparametric}
James~M Hutchinson, Andrew~W Lo, and Tomaso Poggio.
\newblock A nonparametric approach to pricing and hedging derivative securities
  via learning networks.
\newblock \emph{The journal of Finance}, 49\penalty0 (3):\penalty0 851--889,
  1994.

\bibitem[Goodfellow et~al.(2014)Goodfellow, Pouget-Abadie, Mirza, Xu,
  Warde-Farley, Ozair, Courville, and Bengio]{goodfellow2014generative}
Ian Goodfellow, Jean Pouget-Abadie, Mehdi Mirza, Bing Xu, David Warde-Farley,
  Sherjil Ozair, Aaron Courville, and Yoshua Bengio.
\newblock Generative adversarial nets.
\newblock \emph{Advances in neural information processing systems}, 27, 2014.

\bibitem[Zhang et~al.(2019)Zhang, Goodfellow, Metaxas, and
  Odena]{zhang2019self}
Han Zhang, Ian Goodfellow, Dimitris Metaxas, and Augustus Odena.
\newblock Self-attention generative adversarial networks.
\newblock In \emph{International conference on machine learning}, pages
  7354--7363. PMLR, 2019.

\bibitem[Mirza and Osindero(2014)]{mirza2014conditional}
Mehdi Mirza and Simon Osindero.
\newblock Conditional generative adversarial nets.
\newblock \emph{arXiv preprint arXiv:1411.1784}, 2014.

\bibitem[Chen et~al.(2020)Chen, Lin, Xie, Lin, Fan, and Xie]{chen2020dggan}
Liangjian Chen, Shih-Yao Lin, Yusheng Xie, Yen-Yu Lin, Wei Fan, and Xiaohui
  Xie.
\newblock Dggan: Depth-image guided generative adversarial networks for
  disentangling rgb and depth images in 3d hand pose estimation.
\newblock In \emph{Proceedings of the IEEE/CVF Winter Conference on
  Applications of Computer Vision}, pages 411--419, 2020.

\bibitem[Arjovsky et~al.(2017)Arjovsky, Chintala, and
  Bottou]{arjovsky2017wasserstein}
Martin Arjovsky, Soumith Chintala, and L{\'e}on Bottou.
\newblock Wasserstein generative adversarial networks.
\newblock In \emph{International conference on machine learning}, pages
  214--223. PMLR, 2017.

\bibitem[Chen et~al.(2016)Chen, Duan, Houthooft, Schulman, Sutskever, and
  Abbeel]{chen2016infogan}
Xi~Chen, Yan Duan, Rein Houthooft, John Schulman, Ilya Sutskever, and Pieter
  Abbeel.
\newblock Infogan: Interpretable representation learning by information
  maximizing generative adversarial nets.
\newblock \emph{Advances in neural information processing systems}, 29, 2016.

\bibitem[Boyle(1977)]{boyle1977options}
Phelim~P Boyle.
\newblock Options: A monte carlo approach.
\newblock \emph{Journal of financial economics}, 4\penalty0 (3):\penalty0
  323--338, 1977.

\bibitem[Fu and Hu(1995)]{fu1995sensitivity}
Michael~C Fu and Jian-Qlang Hu.
\newblock Sensitivity analysis for monte carlo simulation of option pricing.
\newblock \emph{Probability in the Engineering and Informational Sciences},
  9\penalty0 (3):\penalty0 417--446, 1995.

\bibitem[Birge(1995)]{birge1995quasi}
John~R Birge.
\newblock Quasi-monte carlo approaches to option pricing.
\newblock Technical report, 1995.

\bibitem[Broadie et~al.(1997)Broadie, Glasserman, and
  Jain]{broadie1997enhanced}
Mark Broadie, Paul Glasserman, and Gautam Jain.
\newblock Enhanced monte carlo estimates for american option prices.
\newblock \emph{Journal of Derivatives}, 5:\penalty0 25--44, 1997.

\bibitem[Poirot and Tankov(2006)]{poirot2006monte}
J{\'e}r{\'e}my Poirot and Peter Tankov.
\newblock Monte carlo option pricing for tempered stable (cgmy) processes.
\newblock \emph{Asia-Pacific Financial Markets}, 13\penalty0 (4):\penalty0
  327--344, 2006.

\bibitem[Kim(2022)]{kim2022transferable}
Yo-whan Kim.
\newblock \emph{How Transferable are Video Representations Based on Synthetic
  Data?}
\newblock PhD thesis, Massachusetts Institute of Technology, 2022.

\bibitem[Rosenblatt(1961)]{rosenblatt1961principles}
Frank Rosenblatt.
\newblock Principles of neurodynamics. perceptrons and the theory of brain
  mechanisms.
\newblock Technical report, Cornell Aeronautical Lab Inc Buffalo NY, 1961.

\bibitem[Cassisi et~al.(2012)Cassisi, Montalto, Aliotta, Cannata, and
  Pulvirenti]{cassisi2012similarity}
Carmelo Cassisi, Placido Montalto, Marco Aliotta, Andrea Cannata, and Alfredo
  Pulvirenti.
\newblock Similarity measures and dimensionality reduction techniques for time
  series data mining.
\newblock \emph{Advances in data mining knowledge discovery and applications},
  pages 71--96, 2012.

\bibitem[Quail and Overdahl(2009)]{quail2009financial}
Rob Quail and James~A Overdahl.
\newblock \emph{Financial derivatives: pricing and risk management}, volume~5.
\newblock John Wiley \& Sons, 2009.

\bibitem[Black and Scholes(2019)]{black2019pricing}
Fischer Black and Myron Scholes.
\newblock The pricing of options and corporate liabilities.
\newblock In \emph{World Scientific Reference on Contingent Claims Analysis in
  Corporate Finance: Volume 1: Foundations of CCA and Equity Valuation}, pages
  3--21. World Scientific, 2019.

\bibitem[Poggio and Girosi(1990)]{poggio1990networks}
Tomaso Poggio and Federico Girosi.
\newblock Networks for approximation and learning.
\newblock \emph{Proceedings of the IEEE}, 78\penalty0 (9):\penalty0 1481--1497,
  1990.

\bibitem[Rumelhart et~al.(1985)Rumelhart, Hinton, and
  Williams]{rumelhart1985learning}
David~E Rumelhart, Geoffrey~E Hinton, and Ronald~J Williams.
\newblock Learning internal representations by error propagation.
\newblock Technical report, California Univ San Diego La Jolla Inst for
  Cognitive Science, 1985.

\bibitem[Friedman and Stuetzle(1981)]{friedman1981projection}
Jerome~H Friedman and Werner Stuetzle.
\newblock Projection pursuit regression.
\newblock \emph{Journal of the American statistical Association}, 76\penalty0
  (376):\penalty0 817--823, 1981.

\bibitem[Brewer et~al.(2012)Brewer, Feng, and Kwan]{brewer2012geometric}
Kevin~D Brewer, Yi~Feng, and Clarence~CY Kwan.
\newblock Geometric brownian motion, option pricing, and simulation: Some
  spreadsheet-based exercises in financial modeling.
\newblock \emph{Spreadsheets in Education}, 5\penalty0 (3):\penalty0 4598,
  2012.

\end{thebibliography}

\end{document}